\documentclass[10pt, conference, compsocconf]{IEEEtran}
\ifCLASSINFOpdf
  \usepackage[pdftex]{graphicx}
  \graphicspath{../figures}
  \DeclareGraphicsExtensions{.pdf,.jpeg,.png}
\else
\fi
\usepackage[cmex10]{amsmath}
\usepackage{amsmath}
\usepackage{amssymb}
\usepackage{amsthm}
\newtheorem{theorem}{Theorem}[section]
\newtheorem{lemma}[theorem]{Lemma}
\newtheorem{corollary}[theorem]{Corollary}
\theoremstyle{definition}
\newtheorem{defn}[theorem]{Definition}

\begin{document}
\title{Constructing Adjacency Arrays from Incidence Arrays}


%\author{\IEEEauthorblockN{Hayden Jananthan}
%\IEEEauthorblockA{Vanderbilt University\\
%Nashville, United States\\
%hayden.r.jananthan@vanderbilt.edu}
%\and
%\IEEEauthorblockN{Jeremy Kepner}
%\IEEEauthorblockA{line 1 (of Affiliation): dept. name of organization\\
%line 2: name of organization, acronyms acceptable\\
%line 3: City, Country\\
%line 4: Email: name@xyz.com}
%}

\author{Hayden Jananthan$^{1,2}$ ~~ Karia Dibert$^{2,3}$ ~~ Jeremy Kepner$^{2,3,4}$ \\
\\
$^1$Vanderbilt University Mathematics Department, $^2$MIT Lincoln Laboratory Supercomputing Center, \\
$^3$MIT Mathematics Department, $^4$MIT Computer Science \& AI Laboratory}


%Telephone: (800) 555--1212, Fax: (888) 555--1212}


\maketitle

\begin{abstract}
%The abstract goes here. DO NOT USE SPECIAL CHARACTERS, SYMBOLS, OR MATH IN YOUR TITLE OR ABSTRACT.

Graph construction, a fundamental operation in a data processing pipeline, is typically done by multiplying the incidence array representations of a graph, $\mathbf{E}_\mathrm{in}$ and $\mathbf{E}_\mathrm{out}$, to produce an adjacency array of the graph, $\mathbf{A}$, that can be processed with a variety of algorithms.  This paper provides the mathematical criteria to determine if the product $\mathbf{A} = \mathbf{E}^{\sf T}_\mathrm{out}\mathbf{E}_\mathrm{in}$ will have the required structure of the adjacency array of the graph.  The values in the resulting adjacency array are determined by the corresponding addition $\oplus$ and multiplication $\otimes$ operations used to perform the array multiplication.  Illustrations of the various results possible from different  $\oplus$ and $\otimes$ operations are provided using a small collection of popular music metadata.

\end{abstract}


% \fi
\IEEEpeerreviewmaketitle

\section{Introduction}
\let\thefootnote\relax\footnotetext{This material is based in part upon work supported by the NSF under grant number DMS-1312831.  Any opinions, findings, and conclusions or recommendations expressed in this material are those of the authors and do not necessarily reflect the views of the National Science Foundation.}

 The duality between the canonical representation of graphs as abstract collections of vertices and edges and a matrix representation has been a part of graph theory since its inception \cite{Konig1931,Konig1936}.  Matrix algebra has been recognized as a useful tool in graph theory for nearly as long \cite{Harary1969,Sabadusi1960,Weischel1962,McAndrew1963, TehYap1964,McAndrew1965,HararyTauth1966,Brualdi1967}.  The modern description of the duality between graph algorithms and matrix mathematics (or sparse linear algebra) has been extensively covered in the recent literature \cite{KepnerGilbert2011} and has further spawned the development of the GraphBLAS math library standard (GraphBLAS.org)\cite{Mattson2013} that has been developed in a series of proceedings \cite{Mattson2014a,Mattson2014b,Mattson2015,Buluc2015,Mattson2016} and implementations \cite{BulucGilbert2011,Kepner2012,Ekanadham2014,Hutchison2015,Anderson2016,Zhang2016}.
  
  Adjacency arrays, typically denoted $\mathbf{A}$, have much in common with adjacency matrices.  Likewise, incidence arrays or edge arrays, typically denoted $\mathbf{E}$, have much in common with incidence matrices \cite{BruckRyser1949,FordFulkerson1962,FulkersonGross1965,FisherWing1965}, edge matrices \cite{DobrjanskyjFreudenstein1967}, adjacency lists \cite{BodinKursh1979}, and adjacency structures \cite{Tarjan1972}.  The powerful link between adjacency arrays and incidence arrays via array multiplication is the focus of the first part of this paper.

Incidence arrays are often readily obtained from raw data. In many cases, an associative array representing a spreadsheet or database table is already in the form of an incidence array.  However, to analyze a graph, it is often convenient to represent the graph as an adjacency array.  
%Often these are cases in which $\mathbf{E}_\mathrm{out}=\mathbf{E}_\mathrm{in}$.
Constructing an adjacency array from data stored in an incidence array via array multiplication is one of the most common and important steps in a data processing system.

Given a graph $G$ with vertex set $K_\mathrm{out}\cup K_\mathrm{in}$ and edge set $K$, the construction of adjacency arrays for $G$ relies on the assumption that $\mathbf{E}^{\sf T}_\mathrm{out}\mathbf{E}_\mathrm{in}$ is an adjacency array of $G$.  This assumption is certainly true in the most common case where the value set is composed of non-negative reals and the operations $\oplus$ and $\otimes$ are arithmetic plus ($+$) and arithmetic times (${\times}$) respectively. However, one hallmark of associative arrays is their ability to contain as values nontraditional data. For these value sets, $\oplus$ and $\otimes$ may be redefined to operate on non-numerical values.  For example, for the value of all alphanumeric strings, with
\begin{eqnarray*}
  \oplus &=& \max() \\
  \otimes &=& \min()
\end{eqnarray*}
it is not immediately apparent in this case whether $\mathbf{E}^{\sf T}_\mathrm{out}\mathbf{E}_\mathrm{in}$ is an adjacency array of the graph whose set of vertices is $K_\mathrm{out} \cup K_\mathrm{in}$. In the subsequent sections,  the criteria on the value set $V$ and the operations $\oplus$ and $\otimes$ are presented so that
$$
  \mathbf{A} = \mathbf{E}^{\sf T}_\mathrm{out}\mathbf{E}_\mathrm{in}
$$
always produces an adjacency array \cite{Dibert2015}.

\subsection{Definitions}

For a directed graph (from here onwards, just `graph') $G$, $K_\mathrm{out}$ will denote the set of vertices which are the sources of edges, $K_\mathrm{in}$ will denote the set of vertices which are the targets of edges, and $K$ will denote the set of edges.  The vertex set of $G$ will be assumed to be $K_\mathrm{out} \cup K_\mathrm{in}$. $K_\mathrm{out}$, $K_\mathrm{in}$, and $K$ are assumed to be finite and totally-ordered.

$V$ will denote the set of values that the data can take on, such as non-negative real numbers or the elements of an ordered set.  $\oplus$ and $\otimes$ are binary operations on $V$ (in particular, $V$ is closed under the operations $\oplus$ and $\otimes$), such as $\oplus = +$ and $\otimes = \times$ or $\oplus = \max$ and $\otimes = +$.  $\oplus$ and $\otimes$ each have identity elements $0$ and $1$, respectively, i.e.
\begin{align*}
v \oplus 0 & = 0\oplus v = v \\
v \otimes 1 & = 1\otimes v = v
\end{align*}
for all $v\in V$.  

For the purposes of understanding what algebraic properties are required for $\mathbf{E}^{\sf T}_\mathrm{out}\mathbf{E}_\mathrm{in}$ to be an adjacency array of a graph, $\oplus$ and $\otimes$ will not be assumed to be associative or commutative, and $\otimes$ does not necessarily distribute over $\oplus$, nor is $0$ assumed to be an annihilator of $\otimes$.

\begin{defn}[Associative Array]
An \emph{associative array} is a map $\mathbf{A}: K_1{\times} K_2 \to V$, where $K_1$ and $K_2$ are finite totally-ordered sets, referred to as \emph{key sets} and whose elements are called \emph{keys}, and $V$ is the value set.
\end{defn}

\begin{defn}[Transpose]
If $\mathbf{A}: K_1{\times}K_2 \to V$ is an associative array, then $\mathbf{A}^{\sf T}: K_2{\times} K_1 \to V$ is the associative array defined as
\[
\mathbf{A}^{\sf T}(k_2,k_1) = \mathbf{A}(k_1,k_2)
\]
where $k_1\in K_1$ and $k_2\in K_2$.
\end{defn}

\index{array!multiplication}
\begin{defn}[Array Multiplication]
Multiplication of associative arrays is defined as
$$
  \mathbf{C} = \mathbf{A} {\oplus}.{\otimes} \mathbf{B} = \mathbf{AB}
$$
or more specifically 
$$
  \mathbf{C}(k_1,k_2) =
    \bigoplus\limits_{k_3} \mathbf{A}(k_1,k_3) \otimes \mathbf{B}(k_3,k_2)
$$
where $\mathbf{A}$, $\mathbf{B}$, and $\mathbf{C}$ are associative arrays
\begin{eqnarray*}
  \mathbf{A} : K_1 \times K_3 \rightarrow V \\
  \mathbf{B} : K_3 \times K_1 \rightarrow V \\
  \mathbf{C} : K_1 \times K_2 \rightarrow V
\end{eqnarray*}
and $k_1 \in K_1$, $k_2 \in K_2$, $k_3 \in K_3$.
\end{defn}

\begin{defn}[Incidence Arrays]
If $G$ is a graph with vertex set $K_\mathrm{out} \cup K_\mathrm{in}$ and edge set $K$, then
\begin{description}
\item[$\mathbf{E}_\mathrm{out}$]: $K\times K_\mathrm{out} \to V$ is a \emph{source incidence array} if $\mathbf{E}_\mathrm{out}(k,a) \neq 0$ if and only if the edge $k\in K$ is directed outward from the vertex $a\in K_\mathrm{out}$
\item[$\mathbf{E}_\mathrm{in}$]: $K\times K_\mathrm{in} \to V$ is a \emph{target incidence array} if $\mathbf{E}_\mathrm{in}(k,a) \neq 0$ if and only if the edge $k\in K$ is directed into the vertex $a\in K_\mathrm{in}$.
\end{description}
\end{defn}

\begin{defn}[Adjacency Array]
If $G$ is a graph with vertex set $K_\mathrm{out} \cup K_\mathrm{in}$ and edge set $K$, then $\mathbf{A}: K_\mathrm{out}{\times} K_\mathrm{in} \to V$ is a \emph{adjacency array} if $\mathbf{A}(a,b) \neq 0$ if and only if there is an edge with source $a$ and target $b$.
\end{defn}

\section{Adjacency Array Construction}

If $\mathbf{A}$ is an adjacency array for a graph $G=(K_\mathrm{out}\cup K_\mathrm{in},K)$, then $\mathbf{A}(a,b)\neq 0$ if and only if there is an edge $k$ with source $a$ and target $b$, i.e. so that $\mathbf{E}_\mathrm{out}(k,a) \neq 0$ and $\mathbf{E}_\mathrm{in}(k,a)\neq 0$.  In the case where the product of two non-zero values is non-zero, this can be subsumed to say that $\mathbf{A}(a,b) \neq 0$ if and only if $\mathbf{E}_\mathrm{out}(k,a)\mathbf{E}_\mathrm{in}(k,a)$.  Writing this as
\[		\mathbf{E}_\mathrm{out}(k,a) \mathbf{E}_\mathrm{in}(k,a) = \mathbf{E}^{\sf T}_\mathrm{out}(a,k) \mathbf{E}_\mathrm{in}(k,a) \]
This latter expression looks like a term in the evaluation
\[ (\mathbf{E}_\mathrm{out}^{\sf T}\mathbf{E}_\mathrm{in})(a,b) = \bigoplus_{k\in K}{\mathbf{E}^{\sf T}_\mathrm{out}(a,k) \mathbf{E}_\mathrm{in}(k,b)} \]
but the introduction of more terms means that more assumptions need to be made about the relationships between $\oplus,\otimes$, and $0$.

%%%%%%%%%%%%%%PROBLEM STATEMENT%%%%%%%%%%%%%%%%%%%%%%%%%%%

\begin{theorem} \label{thm:graph construction}
Let $V$ be a set with closed binary operations $\oplus,\otimes$ with identities $0,1\in V$.  Then the following are equivalent:
\begin{enumerate}
\item $\oplus$ and $\otimes$ satisfy the properties
	\begin{enumerate}[(a)]
		\item Zero-Sum-Free: $a\oplus b=0$ if and only if $a=b=0$,
		\item No Zero Divisors: $a\otimes b = 0$ if and only if $a=0$ or $b=0$, and 
		\item $0$ is Annihilator for $\otimes$: $a\otimes 0 = 0\otimes a=0$.
	\end{enumerate}
\item If $G$ is a graph with out-vertex and in-vertex incidence arrays $\mathbf{E}_\mathrm{out}:K{\times} K_\mathrm{out} \rightarrow V$ and $\mathbf{E}_\mathrm{in}: K{\times} K_\mathrm{out} \rightarrow V$, then $\mathbf{E}_\mathrm{out}^{\sf T}\mathbf{E}_\mathrm{in}$ is an adjacency array for $G$.
\end{enumerate}
\end{theorem}
\begin{proof}
%%%%%%%%%%%%%%BACKGROUND%%%%%%%%%%%%%%%%%%%%%%%%%%%%%%%
Let $\mathbf{A} = \mathbf{E}_\mathrm{out}^{\sf T}\mathbf{E}_\mathrm{in}$.  

As above, for $\mathbf{A}$ to be the adjacency array of $G$, the entry $\mathbf{A}(k_\mathrm{out},k_\mathrm{in})$ must be nonzero if and only if there is an edge from $k_\mathrm{out}$ to $k_\mathrm{in}$, which is equivalent to saying that the entry must be nonzero if and only if there is a $k \in K$ such that
\begin{eqnarray*}
  \mathbf{E}^{\sf T}_\mathrm{out}(k_\mathrm{out},k) \neq 0 \\
  \mathbf{E}_\mathrm{in}(k,k_\mathrm{in}) \neq 0
\end{eqnarray*}
Taken altogether, the above pair of equations imply
\begin{multline*} 
\bigoplus_{k \in K} \mathbf{E}^{\sf T}_\mathrm{out}(k_\mathrm{out},k) \otimes \mathbf{E}_\mathrm{in}(k,k_\mathrm{in})  \neq 0 \\
\iff
\exists k\in K \text{ so that } \mathbf{E}^{\sf T}_\mathrm{out}(k_\mathrm{out},k) \neq 0  ~~ \text{and} ~~ \mathbf{E}_\mathrm{in}(k,k_\mathrm{in}) \neq 0 \label{problem}
\end{multline*}

%%%%%%%%%%%%%%SUFFICIENCY%%%%%%%%%%%%%%%%%%%%%%%%%%%%%%%
First, the above condition can be restated in a form that more easily provides the zero-sum-freeness of $\oplus$, lack of zero-divisors for $\otimes$, and the fact that $0$ annihilates under $\otimes$.
Equation~\ref{problem} is equivalent to 
\begin{multline}
  \bigoplus\limits_{k \in K} \mathbf{E}_\mathrm{out}(k,x) \otimes \mathbf{E}_\mathrm{in}(k,y) = 0 \iff \\ 
\nexists k \in K \, \mbox{so that} \, \mathbf{E}_\mathrm{out}(k,x) \neq 0 \text{ and } \mathbf{E}_\mathrm{in}(k,y) \neq 0
\end{multline}
which in turn is equivalent to
\begin{multline}
\bigoplus\limits_{k \in K} \mathbf{E}_\mathrm{out}(k,x) \otimes \mathbf{E}_\mathrm{in}(k,y) = 0 \iff \\
\forall k \in K, \mathbf{E}_\mathrm{out}(k,x) = 0 \text{ or } \mathbf{E}_\mathrm{in}(k,y) = 0
\end{multline}
This expression may be split up into two conditional statements
\begin{multline} 
\bigoplus\limits_{k \in K} \mathbf{E}_\mathrm{out}(k,x) \otimes \mathbf{E}_\mathrm{in}(k,y) = 0 \Rightarrow \\
 \forall k \in K, \mathbf{E}_\mathrm{out}(k,x) = 0 \text{ or }  \mathbf{E}_\mathrm{in}(k,y) = 0 \label{leftright}
 \end{multline}
and
\begin{multline} 
\forall k \in K, \mathbf{E}_\mathrm{out}(k,x) = 0 \, \mbox{or} \, \mathbf{E}_\mathrm{in}(k,y) = 0 \Rightarrow \\
 \bigoplus\limits_{k \in K} \mathbf{E}_\mathrm{out}(k,x) \otimes \mathbf{E}_\mathrm{in}(k,y) = 0 \label{rightleft}
 \end{multline}
%By creating several graphs and using (ii) the necessary properties of $\oplus,\otimes$ are shown:
%%%%%%%%Zero-sum-free%%%%%%%%%
\begin{lemma} \label{zero-sum-free is necessary}
Equation~\ref{leftright} implies that $V$ is zero-sum-free.
\end{lemma}
\begin{proof}
Suppose there exist nonzero $v,w \in V$ such that $v \oplus w = 0$, or that nontrivial additive inverses exist.  Then it is possible to choose a graph $G$ to have edge set $\{k_1,k_2\}$ and vertex set $\{a,b\}$, where both $k_1,k_2$ start from $a$ and end at $b$.  Then defining 
\begin{eqnarray*}
  \mathbf{E}_\mathrm{out}(k_1,a)=v \\
  \mathbf{E}_\mathrm{out}(k_2,a)=w \\
  \mathbf{E}_\mathrm{in}(k_i,b)=1
\end{eqnarray*}
provides proper out-vertex and in-vertex incidence arrays for $G$.  Moreover, it is the case that
$$
\mathbf{E}^{\sf T}_\mathrm{out}\mathbf{E}_\mathrm{in}(b,a)=(v\otimes 1)\oplus (w\otimes 1)=v\oplus w=0
$$
which contradicts Equation~\ref{leftright}.  Therefore, no such nonzero $v$ and $w$ may be present in $V$, meaning it is necessary that $V$ be zero-sum-free.
\end{proof}

%%%%%%%no zero-divisors%%%%%%%%%%%
\begin{lemma} \label{no zero-divisors is necessary}
Equation~\ref{leftright} implies that $V$ has no zero-divisors.
\end{lemma}
\begin{proof}
Suppose $v\otimes w = 0$. Define the graph $G$ to have edge set $\{k\}$ and vertex set $\{a\}$ with a single self-loop given by $k$.  Then define
\begin{eqnarray*}
  \mathbf{E}_\mathrm{out}(k,a)=v \\
   \mathbf{E}_\mathrm{in}(k,a)=w
\end{eqnarray*}
to obtain out-vertex and in-vertex incidence arrays for $G$.  Then
$$
\mathbf{E}^{\sf T}_\mathrm{out}\mathbf{E}_\mathrm{in}(a,a)=\mathbf{E}_\mathrm{out}(k,a)\otimes \mathbf{E}_\mathrm{in}(k,a)=v\otimes w= 0
$$
Thus, Equation~\ref{leftright} implies that $v=w=0$, and hence $V$ has no zero-divisors.
\end{proof}

%%%%%%%zero annihilates %%%%%%%%%%%
\begin{lemma} \label{annihilating is necessary}
Equation~\ref{leftright} implies that $0$ annihilates $V$ under $\otimes$.
\end{lemma}
\begin{proof}
Suppose $v\in V$. Define the graph $G$ to have edge set $\{k_1,k_2\}$ and vertex set $\{a,b\}$, with self-loops at $a$ and $b$ given by $k_1$ and $k_2$, respectively.  Define
$$
  \mathbf{E}_\mathrm{out}(k_1,a)=v=\mathbf{E}_\mathrm{in}(k_1,a)
$$
and
$$
  \mathbf{E}_\mathrm{out}(k_2,b) = v = \mathbf{E}_\mathrm{in}(k_2,b)
$$
(and all other entries in $\mathbf{E}_\mathrm{out}$ and $\mathbf{E}_\mathrm{in}$ equal to $0$) 
results in out-vertex and in-vertex incidence arrays of $G$.  Moreover, it is true that
\begin{eqnarray*}
0 &=& \mathbf{E}^{\sf T}_\mathrm{out}\mathbf{E}_\mathrm{in}(a,b) \\
  &=& \mathbf{E}_\mathrm{out}(k_1,a)\otimes \mathbf{E}_\mathrm{in}(k_1,b)\oplus \mathbf{E}_\mathrm{out}(k_2,a) \otimes \mathbf{E}_\mathrm{in}(k_2,b) \\
  &=&(v\otimes 0) \oplus (0\otimes v)
\end{eqnarray*}
By Lemma~\ref{zero-sum-free is necessary}, $V$ is zero-sum-free so it follows that $v\otimes 0 = 0\otimes v = 0$.  Thus, $0$ is an annihilator for $\otimes$.
\end{proof}

%%%%%%%%%%%%%SUFFICIENT%%%%%%%%%%%%%%%%%%%%%%%%%%%%

Now Theorem~\ref{thm:graph construction}(i) is shown to be sufficient for Theorem~\ref{thm:graph construction}(ii) to hold. Assume that zero is an annihilator, $V$ is zero-sum-free, and $V$ has no zero-divisors.  Zero-sum-freeness and the nonexistence of zero divisors give
\begin{multline}
\exists k \in K \text{ so that } \mathbf{E}_\mathrm{out}(k,x) \neq 0 \text{ and } \mathbf{E}_\mathrm{in}(k,y) \neq 0 \Rightarrow \\
\bigoplus\limits_{k \in K} \mathbf{E}_\mathrm{out}(k,x) \otimes \mathbf{E}_\mathrm{in}(k,y) \neq 0 
\end{multline}
which is the contrapositive of Equation~\ref{leftright}. And, that zero is an annihilator gives
\begin{multline}
\forall k \in K, \mathbf{E}_\mathrm{out}(e,x) = 0 \text{ or } \mathbf{E}_\mathrm{in}(e,y) = 0 \Rightarrow \\
\bigoplus\limits_{k \in } \mathbf{E}_\mathrm{out}(k,x) \otimes \mathbf{E}_\mathrm{in}(k,y) = 0 
\end{multline}
which is (\ref{rightleft}). 
As Equation~\ref{leftright} and Equation~\ref{rightleft} combine to form Equation~\ref{problem}, it is established that the conditions are sufficient for Equation~\ref{problem}.
\end{proof}

\section{Adjacency Array of Reverse Graph}

%%%%%%%%%%%%%%VARIATIONS%%%%%%%%%%%%%%%%%%%%%%%%%%%

The remaining product of the incidence arrays that is defined is $\mathbf{E}^{\sf T}_\mathrm{in}\mathbf{E}_\mathrm{out}$. The above requirements will now be shown to be necessary and sufficient for the remaining product to be the adjacency array of the reverse of the graph. Recall that the reverse of $G$ is the graph $\bar{G}$ in which all the arrows in $G$ have been reversed. Let $G$ be a graph with incidence matrices $\mathbf{E}_\mathrm{out}$ and $\mathbf{E}_\mathrm{in}$.
\begin{corollary}
Condition (i) in Theorem~\ref{thm:graph construction} are necessary and sufficient so that $\mathbf{E}^{\sf T}_\mathrm{in}\mathbf{E}_\mathrm{out}$ is an adjacency matrix of the reverse of $G$.
\label{revcor}
\end{corollary}
\begin{proof}
Let $\bar{G}$ denote the reverse of $G$, and let $\bar{\mathbf{E}}_\mathrm{out}$ and $\bar{\mathbf{E}}_\mathrm{in}$ be out-vertex and in-vertex incidence arrays for $\bar{G}$, respectively.  Recall that $\bar{G}$ is defined to have the same edge and vertex sets as $G$ but changes the directions of the edges, in other words, if an edge $k$ leaves a vertex $a$ in $G$, then it enters $a$ in $\bar{G}$, and vice versa.  As such, $\mathbf{E}_\mathrm{out}(k,a) \neq 0$ if and only if $\bar{\mathbf{E}}_\mathrm{in}(k,a) \neq 0$, and likewise $\mathbf{E}_\mathrm{in}(k,a)\neq 0$ if and only if $\bar{\mathbf{E}}_\mathrm{out}(k,a) \neq 0$.  As such, choosing $\mathbf{E}_\mathrm{out}=\bar{\mathbf{E}}_\mathrm{in}$ and $\mathbf{E}_\mathrm{in}=\bar{\mathbf{E}}_\mathrm{out}$ gives valid in-vertex and out-vertex incidence matrices for $\bar{G}$, respectively.    Then by Theorem~\ref{thm:graph construction} it can be shown that
$$
  \bar{\mathbf{E}}^{\sf T}_\mathrm{out}\bar{\mathbf{E}}_\mathrm{in}=\mathbf{E}^{\sf T}_\mathrm{in}\mathbf{E}_\mathrm{out}
$$
\end{proof}

It is now straightforward to identify algebraic structures that comply with the established criteria. Notably, all zero-sum-free semirings with no zero-divisors comply, such as $\mathbb{N}$ or $\mathbb{R}_{\geq 0}$ with the standard addition and multiplication.  In addition, any linearly ordered set with $\oplus$ and $\otimes$ given by $\max$ and $\min$, respectively.  Some non-examples, however, include the max-plus algebra or non-trivial Boolean algebras, which do not satisfy the zero-product property, or rings, which except for the zero ring are not zero-sum-free.  Furthermore, the value sets of associative arrays need not be defined exclusively as semirings, as several semiring-like structures satisfy the criteria. These structures may lack the properties of additive or multiplicative commutativity, additive or multiplicative associativity, or distributivity of multiplication over addition, which are not necessary to ensure that the product of incidence arrays yields an adjacency array.

The criteria guarantee an accurate adjacency array for any dataset that satisfies them, regardless of value distribution in the incidence arrays. However, if the incidence arrays are known to possess a certain structure, it is possible to circumvent some of the conditions and still always produce adjacency arrays.  For example, if each key set of an undirected incidence array $\mathbf{E}$ is a list of documents and the array entries are sets of words shared by documents, then it is necessary that a word in $\mathbf{E}(i,j)$ and $\mathbf{E}(m,n)$ has to be in $\mathbf{E}(i,n)$ and $\mathbf{E}(m,j)$. This structure means that when multiplying $\mathbf{E}^{\sf T}\mathbf{E}$ using $\oplus = \cup$ and $\otimes = \cap$, a nonempty set will never be ``multiplied'' by (intersected with) a disjoint nonempty set. This eliminates the need for the zero-product property to be satisfied, as every multiplication of nonempty sets is already guaranteed to produce a nonempty set. The array produced will contain as entries a list of words shared by those two documents.

Though the criteria ensure that the product of incidence arrays will be an adjacency array, they do not ensure that certain matrix properties hold. For example, the property $(\mathbf{AB})^{\sf T}=\mathbf{B}^{\sf T}\mathbf{A}^{\sf T}$ may be violated under these criteria, as $(\mathbf{E}^{\sf T}_\mathrm{out}\mathbf{E}^{\sf T}_\mathrm{in})$ is not necessarily equal to $\mathbf{E}^{\sf T}_\mathrm{in}\mathbf{E}_\mathrm{out}$. (For this matrix transpose property to always hold, the operation $\otimes$ would have to be commutative.)

\section{Graph Construction with Different Semirings}

   The ability to change $\oplus$ and $\otimes$ operations allows different graph adjacency arrays to be constructed using the same element-wise addition, element-wise multiplication, and array multiplication syntax.  Specific pairs of operations are best suited for constructing certain types of adjacency arrays.  The  pattern of edges resulting from array multiplication of incidence arrays is generally preserved for various semirings.  However, the non-zero values assigned to the edges can be very different and enable the construction different graphs.

For example, constructing an adjacency array of the graph of music writers connected to music genres from Figure~\ref{fig:D4M-SparseAssocArray} begins with selecting the incidence sub-arrays $\mathbf{E}_1$ and $\mathbf{E}_2$ as shown in Figure~\ref{fig:Music-IncidenceArrays}.  Array multiplication of $\mathbf{E}_1^{\sf T}$ with $\mathbf{E}_2$ produces the desired adjacency array of the graph.   Figure~\ref{fig:Semiring-SparseCorrelation} illustrates this array multiplication for different operator pairs $\oplus$ and $\otimes$.
%Semirings are shown in black and non-semirings are shown in grey.

\begin{figure}[htb]
  \centering
    \includegraphics[width=.5\textwidth]{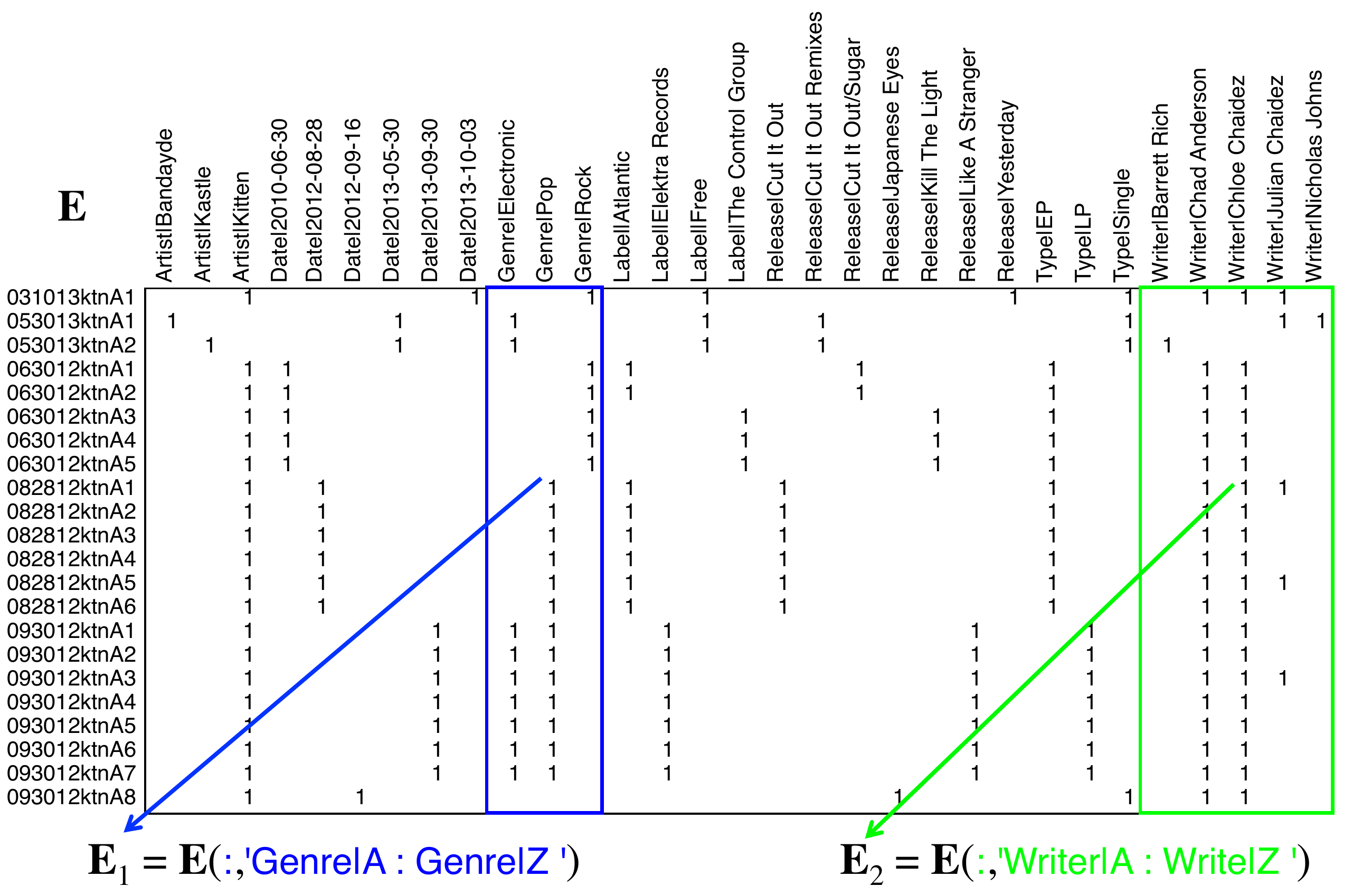}
     \caption{D4M sparse associative array $\mathbf{E}$ representation of a table of data from a music database.  The column key and the value are concatenated with a separator symbol (in this case $|$) resulting in every unique pair of column and value having its own column in the sparse view.  The new value is usually 1 to denote the existence of an entry.  Column keys are an ordered set of database fields.  Sub-arrays  $\mathbf{E}_1$ and $\mathbf{E}_2$ are selected with Matlab-style notation to denote all of the row keys and ranges of column keys.}
     \label{fig:D4M-SparseAssocArray}
\end{figure}

\begin{figure}[htb]
  \centering
    \includegraphics[width=3in]{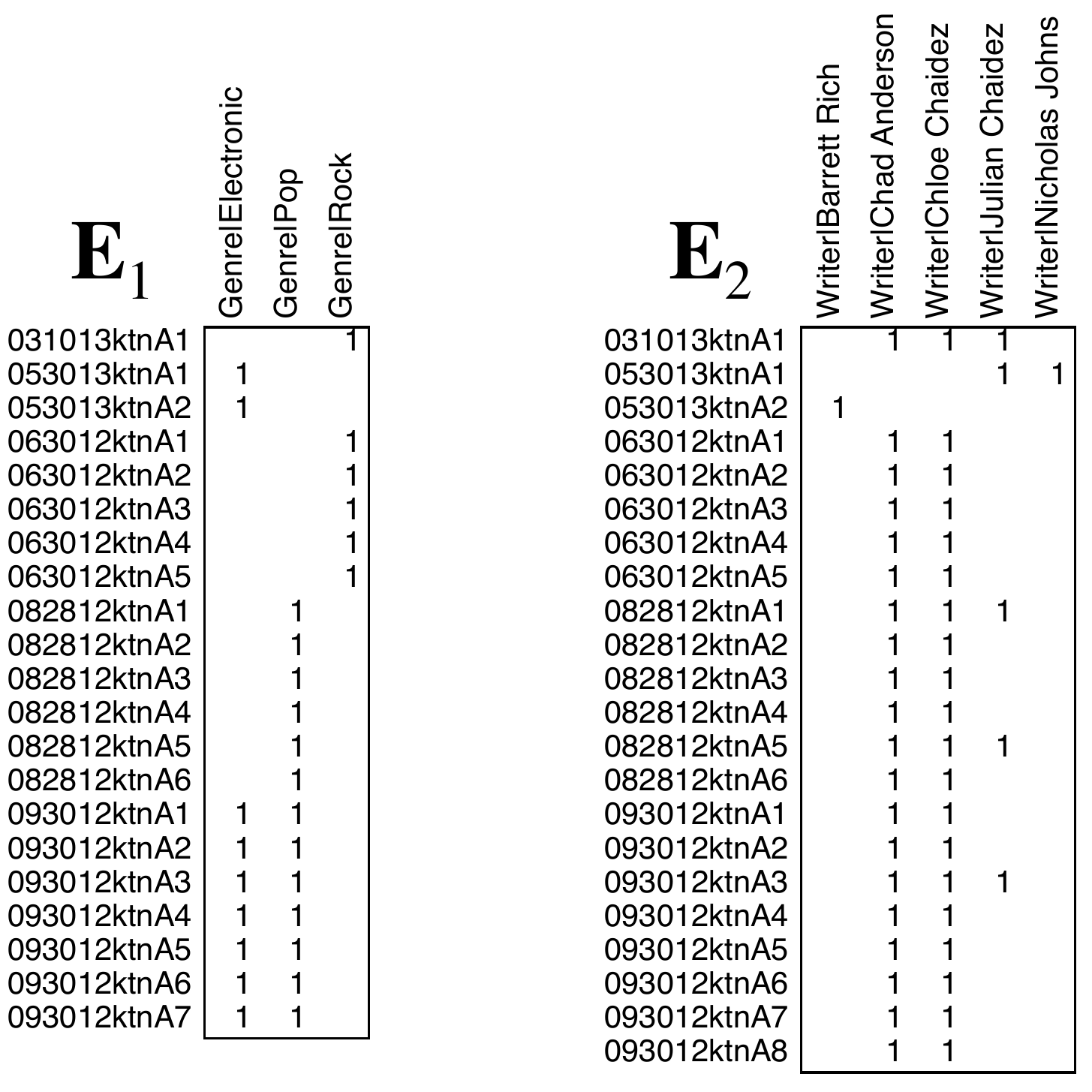}
      \caption{Incidence arrays of music writers and music genres $\mathbf{E}_1$ and $\mathbf{E}_2$ as defined in Figure~\ref{fig:D4M-SparseAssocArray} for different tracks of music.}
      \label{fig:Music-IncidenceArrays}
\end{figure}

The pattern of edges among vertices in the adjacency arrays shown Figure~\ref{fig:Semiring-SparseCorrelation} are the same for the different operator pairs, but the edge weights differ.  All the non-zero values in $\mathbf{E}_1$ and $\mathbf{E}_2$ are 1. All the $\otimes$ operators in Figure~\ref{fig:Semiring-SparseCorrelation} have the property
$$
   0 \otimes 1 = 1 \otimes 0 = 0
$$
for their respective values of zero be it 0, $\text{-}\infty$, or $\infty$.  Likewise, all the $\otimes$ operators in Figure~\ref{fig:Semiring-SparseCorrelation} also have the property
$$
   1 \otimes 1 = 1
$$
except where $\otimes = +$, in which case
$$
   1 \otimes 1 = 2
$$
The differences in the adjacency array weights are less pronounced then if the values of $\mathbf{E}_1$ and $\mathbf{E}_2$ were more diverse.  The most apparent difference is between the ${+}.{\times}$ semiring and the other semirings in Figure~\ref{fig:Semiring-SparseCorrelation}.  In the case of ${+}.{\times}$ semiring, the $\oplus$ operation $+$ aggregates values from all the edges between two vertices. Additional positive edges will increase the overall weight in the adjacency array.  In the other pairs of operations, the $\oplus$ operator is either $\max$ or $\min$, which effectively selects only one edge weight to use for assigning the overall weight.  Additional edges will only impact the edge weight in the adjacency array if the new edge is an appropriate maximum or minimum value.  Thus, ${+}.{\times}$ constructs adjacency arrays that aggregate all the edges.  The sother emirings construct adjacency arrays that select extremal edges.  Each can be useful for construction graph adjacency arrays in appropriate context.

\begin{figure}[htb]
  \centering
    \includegraphics[width=3in]{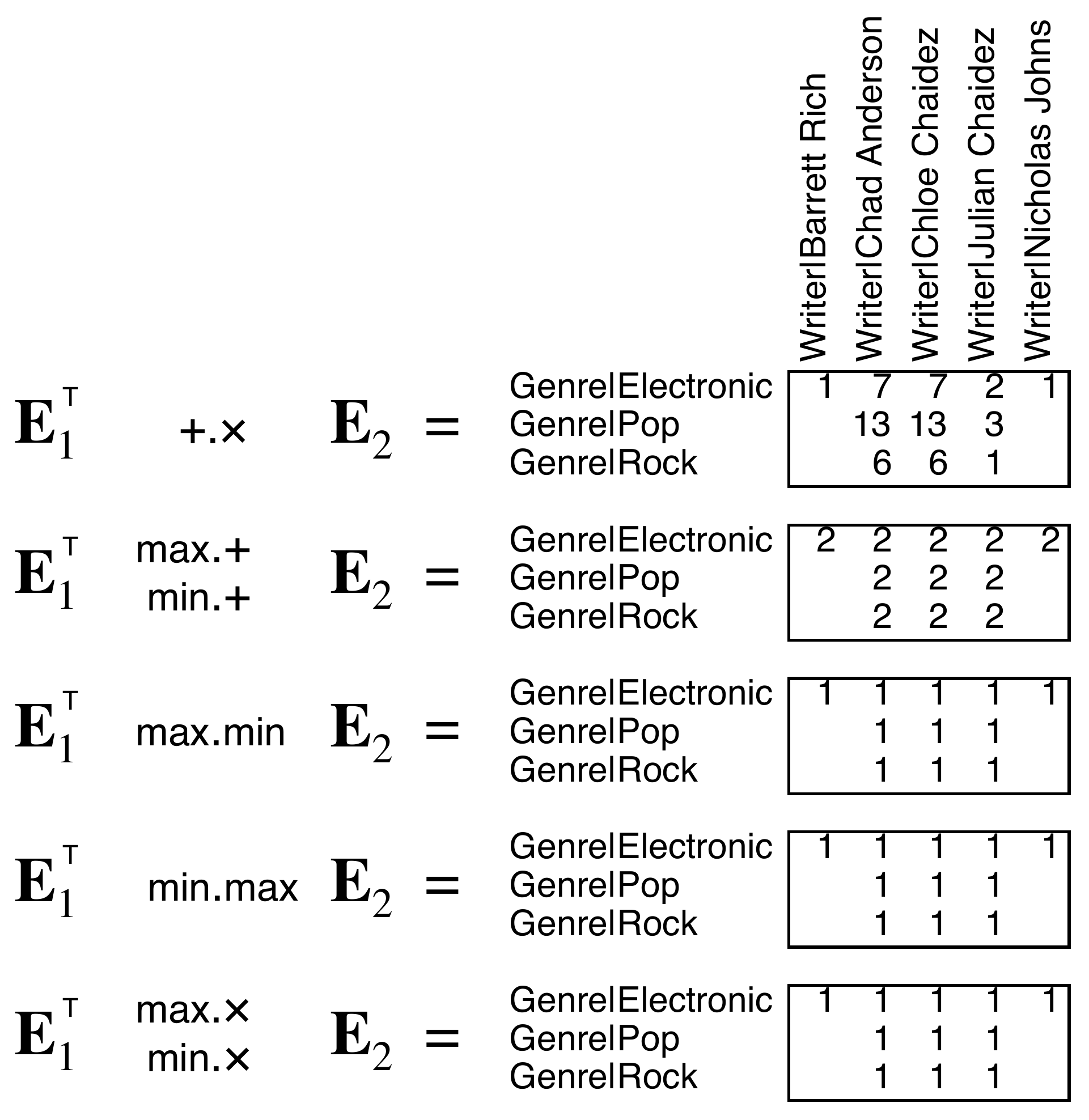}
      \caption{Creating a graph of music writers related to music genres can be computed by multiplying $\mathbf{E}_1$ and $\mathbf{E}_2$ as defined in Figure~\ref{fig:Music-IncidenceArrays}.  This correlation is performed using the  transpose operation ${\sf ^T}$ and the array multiplication operation ${\oplus}.{\otimes}$.  The resulting associative array has row keys taken from the column keys of $\mathbf{E}_1$  and column keys taken from the column keys of $\mathbf{E}_2$.  The values represent the weights on the edges between the vertices of the graph. Different pairs of operations $\oplus$ and $\otimes$ produce different results. For display convenience, operator pairs that produce the same values \underline{\smash{in this specific example}} are stacked.
%Semiring operators pairs are shown in black and non-semiring operator pairs are shown in grey.
}
      \label{fig:Semiring-SparseCorrelation}
\end{figure}

The impact of different semirings on the graph adjacency array weights are more pronounced if the values of $\mathbf{E}_1$ and $\mathbf{E}_2$ are more diverse.  Figure~\ref{fig:Music-IncidenceArrays123} modifies $\mathbf{E}_1$ so that a value of 2 is given to the non-zero values in the column {\sf Genre$|$Pop} and a values of 3 is given to the non-zero values in the column {\sf Genre$|$Rock}.

\begin{figure}[htb]
  \centering
    \includegraphics[width=3in]{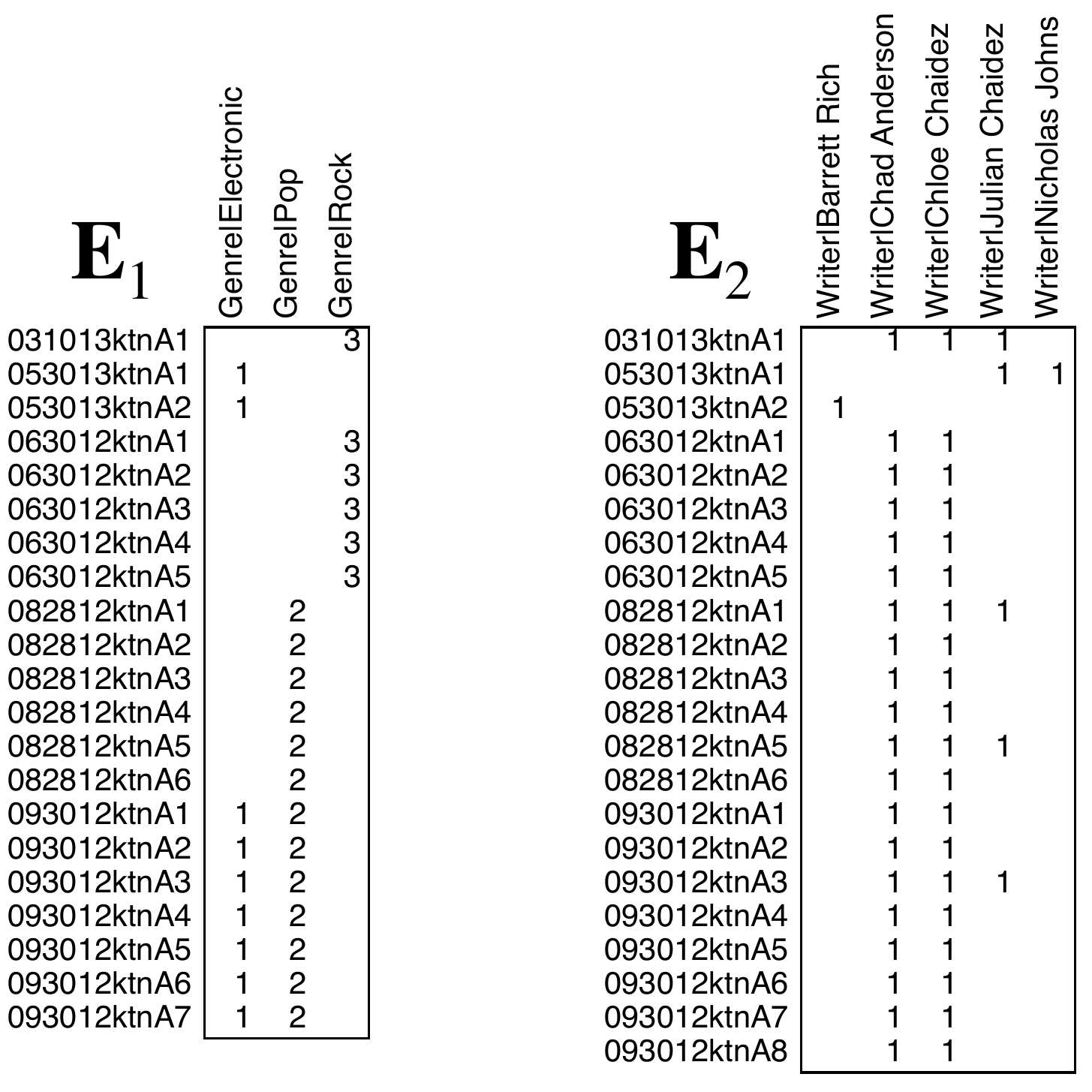}
      \caption{Incidence arrays from Figure~\ref{fig:Music-IncidenceArrays} modified so that the non-zero values of $\mathbf{E}_1$ take on the values 1, 2, and 3.}
      \label{fig:Music-IncidenceArrays123}
\end{figure}

Figure~\ref{fig:Semiring-SparseCorrelation123} shows the results of constructing adjacency arrays with $\mathbf{E}_1$ and $\mathbf{E}_2$ using different semirings.  The impact of changing the values in  $\mathbf{E}_1$ can be seen by comparing Figure~\ref{fig:Semiring-SparseCorrelation} with Figure~\ref{fig:Semiring-SparseCorrelation123}.  For the ${+}.{\times}$ semiring, the values in the adjacency array rows {\sf Genre$|$Pop} and {\sf Genre$|$Rock} are multiplied by 2 and 3.  The increased adjacency array values for these rows are a result of the $\otimes$ operator being arithmetic multiplication $\times$ so that
\begin{eqnarray*}
  2 \otimes 1 = 2 \times 1 = 2 \\
  3 \otimes 1 = 3 \times 1 = 3
\end{eqnarray*}
For the ${\max}.{+}$ and ${\min}.{+}$ semirings, the values in the adjacency array rows {\sf Genre$|$Pop} and {\sf Genre$|$Rock}  are larger by and 1 and 2.  The larger values in the adjacency array of these rows is due to the $\otimes$ operator being arithmetic addition $+$ resulting in
\begin{eqnarray*}
  2 \otimes 1 = 2 + 1 = 3 \\
  3 \otimes 1 = 3 + 1 = 4
\end{eqnarray*}
For the ${\max}.{\min}$ semiring, Figure~\ref{fig:Semiring-SparseCorrelation} and Figure~\ref{fig:Semiring-SparseCorrelation123} have the same adjacency array because $\mathbf{E}_2$ is unchanged. The $\otimes$ operator corresponding to the minimum value function continues to select the smaller non-zero values from $\mathbf{E}_2$
\begin{eqnarray*}
  2 \otimes 1 = \min(2,1) = 1 \\
  3 \otimes 1 = \min(3,1) = 1
\end{eqnarray*}
In contrast, for the ${\min}.{\max}$ semiring, the values in the adjacency array rows {\sf Genre$|$Pop} and {\sf Genre$|$Rock} are larger by and 1 and 2.  The increase in adjacency array values for these rows are a result of the $\otimes$ operator selecting the larger non-zero values from $\mathbf{E}_1$
\begin{eqnarray*}
  2 \otimes 1 = \max(2,1) = 2 \\
  3 \otimes 1 = \max(3,1) = 3
\end{eqnarray*}
Finally, for the ${\max}.{\times}$ and ${\min}.{\times}$ semirings, the values in the adjacency array rows {\sf Genre$|$Pop} and {\sf Genre$|$Rock} are increased by and 1 and 2.  Similar to the ${+}.{\times}$ semiring, the larger adjacency array values for these rows are a result of the $\otimes$ operator being arithmetic multiplication $\times$ resulting in
\begin{eqnarray*}
  2 \otimes 1 = 2 \times 1 = 2 \\
  3 \otimes 1 = 3 \times 1 = 3
\end{eqnarray*}

\begin{figure}[htb]
  \centering
    \includegraphics[width=3in]{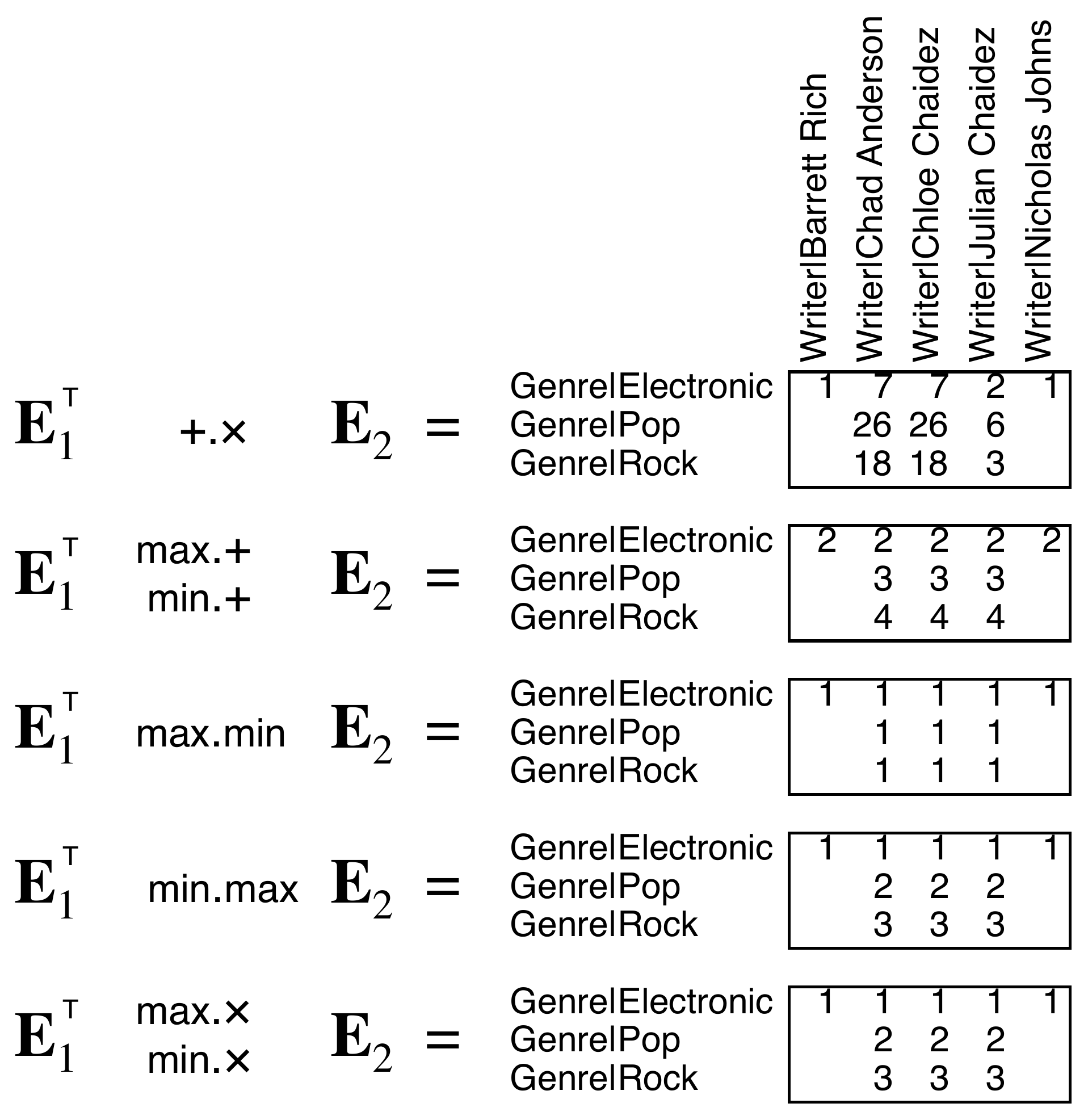}
      \caption{Building a graph of music writers connected with the music genres can be accomplished by multiplying $\mathbf{E}_1$ and $\mathbf{E}_2$ as defined in Figure~\ref{fig:Music-IncidenceArrays123}.  The correlation is computed with the transpose operation ${\sf ^T}$ and array multiplication ${\oplus}.{\otimes}$.  The resulting associative array has row keys taken from the column keys of $\mathbf{E}_1$  and column keys taken from the column keys of $\mathbf{E}_2$.  The values represent the weights on the edges between the vertices of the graph. Different pairs of operations $\oplus$ and $\otimes$ produce different results. For display convenience, operator pairs that produce the same values \underline{\smash{in this specific example}} are stacked.
%Semiring operators pairs are shown in black and non-semiring operator pairs are shown in grey.
}
      \label{fig:Semiring-SparseCorrelation123}
\end{figure}

Figures~\ref{fig:Semiring-SparseCorrelation} and \ref{fig:Semiring-SparseCorrelation123} show that a wide range of graph adjacency arrays can be constructed via array multiplication of incidence arrays over different semirings.  A synopsis of the graph constructions illustrated in Figures~\ref{fig:Semiring-SparseCorrelation} and \ref{fig:Semiring-SparseCorrelation123} is as follows
\begin{description}
\item[${+}.{\times}$] sum of products of edge weights connecting two vertices; computes the strength of all connections between two connected vertices.
\item[${\max}.{\times}$] \quad maximum of products edge weights connecting two vertices; selects the edge with largest weighted product of all the edges connecting two vertices.
\item[${\min}.{\times}$] \quad minimum of products edge weights connecting two vertices; selects the edge with smallest weighted product of all the edges connecting two vertices.
\item[${\max}.{+}$] \quad maximum of sum of edge weights connecting two vertices; selects the edge with largest weighted sum of all the edges connecting two vertices.
\item[${\min}.{+}$] \quad minimum of sum of edge weights connecting two vertices; selects the edge with smallest weighted sum of all the edges connecting two vertices.
\item[${\max}.{\min}$] \quad \quad maximum of the minimum of weights connecting two vertices; selects the largest of all the shortest connections between two vertices.
\item[${\min}.{\max}$] \quad \quad minimum of the maximum of weights connecting two vertices; selects the smallest of all the largest connections between two vertices.
%\item[${\max}.{\max}$] maximum of maximum of weight along each path (not a semiring); computes the longest of all the longest connections between the starting vertex and the ending vertices.
%\item[${\min}.{\min}$] minimum of minimum of weight along each path (not a semiring); computes the shortest of all the shortest connections between the starting vertex and the ending vertices.
\end{description}

\section{Conclusion}

Graph construction, a fundamental operation in a data processing pipeline, is typically done by multiplying the incidence array representations of a graph, $\mathbf{E}_\mathrm{in}$ and $\mathbf{E}_\mathrm{out}$, to produce an adjacency array of the graph, $\mathbf{A}$.  The mathematical criteria to determine if $\mathbf{A}$ will have the required structure of the adjacency array of the graph over are as follows.  Let $V$ be a set with closed binary operations $\oplus,\otimes$ with identities $0,1\in V$.  Then the following are equivalent:
\begin{enumerate}
\item $\oplus$ and $\otimes$ satisfy the properties
	\begin{enumerate}[(a)]
		\item Zero-Sum-Free: $a\oplus b=0$ if and only if $a=b=0$,
		\item No Zero Divisors: $a\otimes b = 0$ if and only if $a=0$ or $b=0$, and 
		\item $0$ is Annihilator for $\otimes$: $a\otimes 0 = 0\otimes a=0$.
	\end{enumerate}
\item If $G$ is a graph with out-vertex and in-vertex incidence arrays $\mathbf{E}_\mathrm{out}:K{\times} K_\mathrm{out} \rightarrow V$ and $\mathbf{E}_\mathrm{in}: K{\times} K_\mathrm{out} \rightarrow V$, then $\mathbf{E}_\mathrm{out}^{\sf T}\mathbf{E}_\mathrm{in}$ is an adjacency array for $G$.
\end{enumerate}
The values in the resulting adjacency array are determined by the corresponding addition $\oplus$ and multiplication $\otimes$ operations used to perform the array multiplication.


\section*{Acknowledgment}

The authors would like to thank Paul Burkhardt, Alan Edelman, Sterling Foster, Vijay Gadepally, Sam Madden, Dave Martinez, Tom Mattson, Albert Reuther, Victor Roytburd, and Michael Stonebraker.



%\bibliographystyle{IEEEtran}

\end{document}